\documentclass[a4paper,12pt]{elsarticle}
\usepackage[T1]{fontenc}
\usepackage[english]{babel}
\usepackage{ae,aecompl}
\usepackage{amsmath, amsthm}
\usepackage{amssymb}
\usepackage[utf8]{inputenc}
\usepackage[section] {placeins}
\usepackage[ruled, vlined, linesnumbered]{algorithm2e}
\newtheorem {Theorem}                 {Theorem}         [section]
\newtheorem {theorem}      [Theorem]  {Theorem}
\newtheorem {myalgorithm}    [Theorem]  {Algorithm}

\newtheorem {lemma}        [Theorem]  {Lemma}

\input amssym.def
\input amssym

\usepackage{pgfplots}
\usepackage{verbatim}
\usepackage{subfigure}
\usepackage{tikz}
\usetikzlibrary{shapes,arrows,backgrounds,mindmap}
\usepackage{titlesec}
\journal{arXiv}

\begin{document}
	\begin{frontmatter}
		\title{Minimum $2$- edge strongly biconnected spanning directed subgraph problem}
		\author{Raed Jaberi}
		
		\begin{abstract} 
			 
	    Wu and Grumbach introduced the concept of strongly biconnected directed graphs. A directed graph $G=(V,E)$ is called strongly biconnected if the directed graph $G$ is strongly connected and the underlying undirected graph of $G$ is biconnected. A strongly biconnected directed graph $G=(V,E)$ is said to be $2$-edge-strongly biconnected if it has at least three vertices and the directed subgraph $(V,E\setminus\left\lbrace e\right\rbrace )$ is strongly biconnected for all $e \in E$.  Let $G=(V,E)$ be a $2$-edge-strongly biconnected directed graph. In this paper we study the problem of computing a minimum size subset $H \subseteq E$ such that the directed subgraph $(V,H)$ is $2$-edge-strongly biconnected.
		\end{abstract} 
	\end{frontmatter}
	\section{Introduction}
	 Wu and Grumbach \cite{WG2010} introduced the concept of strongly biconnected directed graphs, A directed graph $G=(V,E)$ is called strongly biconnected if the directed graph $G$ is strongly connected and the underlying undirected graph of $G$ is biconnected. An edge $e$ in a strongly biconnected directed graph $G=(V,E)$ is b-bridge if the subgraph $(V,E \setminus\left\lbrace e \right\rbrace )$ is not strongly biconnected. 
	 A strongly biconnected directed graph $G=(V,E)$ is called $k$-edge-strongly biconnected if for each $L\subset E$ with $|L|<k$, the  subgraph $(V,E\setminus L)$ is strongly biconnected. A strongly biconnected directed graph $G=(V,E)$ is therefore $2$-edge-strongly biconnected if and only if it has no b-bridges. 		
	Given a $k$-edge-strongly biconnected directed graph $G=(V,E)$, the minimum $k$-edge-strongly biconnected spanning subgraph problem (denoted by MKESBSS) is to compute a minimum size subset $E_{ke} \subseteq E$ such that the directed subgraph $(V,E_{ke})$ is $k$-edge-strongly biconnected. In this paper we study the MKESBSS problem for $k=2$.
 Note that optimal solutions for minimum $2$-edge-connected spanning subgraph  problem are not necessarily feasible solutions for the $2$-edge strongly biconnnected spanning subgraph problem, as shown in Figure \ref{figure:exampleoptimalsol}.

\begin{figure}[htp]
	\centering
	
	\subfigure[]{
	\scalebox{0.79}{
	\begin{tikzpicture}[xscale=2]
	\tikzstyle{every node}=[color=black,draw,circle,minimum size=0.9cm]
	\node (v1) at (-1.6,3.25) {$1$};
	\node (v2) at (-2.5,0) {$2$};
	\node (v3) at (-0.51, -2.5) {$3$};
	\node (v4) at (3.9,-1.5) {$4$};
	\node (v5) at (0.5,2.6) {$5$};
	\node (v6) at (5,1.7) {$6$};
	\node (v7) at (3.7,3.2) {$7$};
	\node (v8) at (4.9,0) {$8$};
	\node (v9) at (-3.43,1) {$9$};
	\node (v10) at (-3.43,3) {$10$};
	\node (v11) at (1,-2.1){$11$};
	\node (v12) at (2.6,-2.1) {$12$};
	
	\begin{scope}   
	\tikzstyle{every node}=[auto=right]   
	
	\draw [-triangle 45] (v9) to (v10);
	\draw [-triangle 45] (v10) to[bend right] (v9);
	\draw [-triangle 45] (v9) to[bend right] (v2);
	\draw [-triangle 45] (v2) to (v9);
	\draw [-triangle 45] (v9) to (v10);
	\draw [-triangle 45] (v10) to[bend right] (v9);
	\draw [-triangle 45] (v1) to [bend right](v10);
	\draw [-triangle 45] (v10) to (v1);
	\draw [-triangle 45] (v1) to (v5);
	\draw [-triangle 45] (v5) to[bend right] (v1);
	\draw [-triangle 45] (v5) to (v7);
	\draw [-triangle 45] (v7) to[bend right] (v5);
	\draw [-triangle 45] (v7) to (v6);
	\draw [-triangle 45] (v6) to[bend right] (v7);
	\draw [-triangle 45] (v8) to (v4);
	\draw [-triangle 45] (v6) to[bend right] (v8);
	\draw [-triangle 45] (v8) to (v6);
	\draw [-triangle 45] (v4) to[bend right] (v8);
	\draw [-triangle 45] (v5) to[bend right] (v3);
	\draw [-triangle 45] (v3) to (v5);
	\draw [-triangle 45] (v3) to (v2);
	\draw [-triangle 45] (v2) to[bend right] (v3);
	\draw [-triangle 45] (v12) to[bend right]  (v4);
	\draw [-triangle 45] (v4) to(v12);
	\draw [-triangle 45] (v12) to (v11);
	\draw [-triangle 45] (v11) to[bend right] (v12);
	\draw [-triangle 45] (v11) to [bend right](v5);
	\draw [-triangle 45] (v5) to (v11);
	\draw [-triangle 45] (v2) to (v10);

		\draw [-triangle 45] (v6) to[bend left] (v10);
	
		\draw [-triangle 45] (v2) to (v4);
	\end{scope}
	\end{tikzpicture}}
	
	}
	\subfigure[]{
	\scalebox{0.79}{
	\begin{tikzpicture}[xscale=2]
	\tikzstyle{every node}=[color=black,draw,circle,minimum size=0.9cm]
	\node (v1) at (-1.6,3.25) {$1$};
	\node (v2) at (-2.5,0) {$2$};
	\node (v3) at (-0.51, -2.5) {$3$};
	\node (v4) at (3.9,-1.5) {$4$};
	\node (v5) at (0.5,2.6) {$5$};
	\node (v6) at (5,1.7) {$6$};
	\node (v7) at (3.7,3.2) {$7$};
	\node (v8) at (4.9,0) {$8$};
	\node (v9) at (-3.43,1) {$9$};
	\node (v10) at (-3.43,3) {$10$};
	\node (v11) at (1,-2.1){$11$};
	\node (v12) at (2.6,-2.1) {$12$};
	
	\begin{scope}   
	\tikzstyle{every node}=[auto=right]   
	
	\draw [-triangle 45] (v9) to (v10);
	\draw [-triangle 45] (v10) to[bend right] (v9);
	\draw [-triangle 45] (v9) to[bend right] (v2);
	\draw [-triangle 45] (v2) to (v9);
	\draw [-triangle 45] (v9) to (v10);
	\draw [-triangle 45] (v10) to[bend right] (v9);
	\draw [-triangle 45] (v1) to [bend right](v10);
	\draw [-triangle 45] (v10) to (v1);
	\draw [-triangle 45] (v1) to (v5);
	\draw [-triangle 45] (v5) to[bend right] (v1);
	\draw [-triangle 45] (v5) to (v7);
	\draw [-triangle 45] (v7) to[bend right] (v5);
	\draw [-triangle 45] (v7) to (v6);
	\draw [-triangle 45] (v6) to[bend right] (v7);
	\draw [-triangle 45] (v8) to (v4);
	\draw [-triangle 45] (v6) to[bend right] (v8);
	\draw [-triangle 45] (v8) to (v6);
	\draw [-triangle 45] (v4) to[bend right] (v8);
	\draw [-triangle 45] (v5) to[bend right] (v3);
	\draw [-triangle 45] (v3) to (v5);
	\draw [-triangle 45] (v3) to (v2);
	\draw [-triangle 45] (v2) to[bend right] (v3);
	\draw [-triangle 45] (v12) to[bend right]  (v4);
	\draw [-triangle 45] (v4) to(v12);
	\draw [-triangle 45] (v12) to (v11);
	\draw [-triangle 45] (v11) to[bend right] (v12);
	\draw [-triangle 45] (v11) to [bend right](v5);
	\draw [-triangle 45] (v5) to (v11);

	\end{scope}
		\end{tikzpicture}}}
\subfigure[]{
	\scalebox{0.79}{
	\begin{tikzpicture}[xscale=2]
	\tikzstyle{every node}=[color=black,draw,circle,minimum size=0.9cm]
	\node (v1) at (-1.6,3.25) {$1$};
	\node (v2) at (-2.5,0) {$2$};
	\node (v3) at (-0.51, -2.5) {$3$};
	\node (v4) at (3.9,-1.5) {$4$};
	\node (v5) at (0.5,2.6) {$5$};
	\node (v6) at (5,1.7) {$6$};
	\node (v7) at (3.7,3.2) {$7$};
	\node (v8) at (4.9,0) {$8$};
	\node (v9) at (-3.43,1) {$9$};
	\node (v10) at (-3.43,3) {$10$};
	\node (v11) at (1,-2.1){$11$};
	\node (v12) at (2.6,-2.1) {$12$};
	
	\begin{scope}   
	\tikzstyle{every node}=[auto=right]   
	
	\draw [-triangle 45] (v9) to (v10);
	\draw [-triangle 45] (v10) to[bend right] (v9);
	\draw [-triangle 45] (v9) to[bend right] (v2);
	\draw [-triangle 45] (v2) to (v9);
	\draw [-triangle 45] (v9) to (v10);
	\draw [-triangle 45] (v10) to[bend right] (v9);
	\draw [-triangle 45] (v1) to [bend right](v10);
	\draw [-triangle 45] (v10) to (v1);
	\draw [-triangle 45] (v1) to (v5);
	\draw [-triangle 45] (v5) to[bend right] (v1);
	\draw [-triangle 45] (v5) to (v7);
	\draw [-triangle 45] (v7) to[bend right] (v5);
	\draw [-triangle 45] (v7) to (v6);
	\draw [-triangle 45] (v6) to[bend right] (v7);
	\draw [-triangle 45] (v8) to (v4);
	\draw [-triangle 45] (v6) to[bend right] (v8);
	\draw [-triangle 45] (v8) to (v6);
	\draw [-triangle 45] (v4) to[bend right] (v8);
	\draw [-triangle 45] (v5) to[bend right] (v3);
	\draw [-triangle 45] (v3) to (v5);
	\draw [-triangle 45] (v3) to (v2);
	\draw [-triangle 45] (v2) to[bend right] (v3);
	\draw [-triangle 45] (v12) to[bend right]  (v4);
	\draw [-triangle 45] (v4) to(v12);
	\draw [-triangle 45] (v12) to (v11);
	\draw [-triangle 45] (v11) to[bend right] (v12);
	\draw [-triangle 45] (v11) to [bend right](v5);
	\draw [-triangle 45] (v5) to (v11);

		\draw [-triangle 45] (v6) to[bend left] (v10);
	
		\draw [-triangle 45] (v2) to (v4);
	\end{scope}
	\end{tikzpicture}}}
\caption{(a) A $2$-edge strongly biconnected directed graph. (b) This subgraph is not $2$-edge strongly biconnected. But it is an optimal solution for the minimum $2$-edge-connected spanning subgraph problem. (c) An optimal solution for the minimum $2$-edge strongly biconnected spanning subgraph problem}
\label{figure:exampleoptimalsol}
\end{figure}
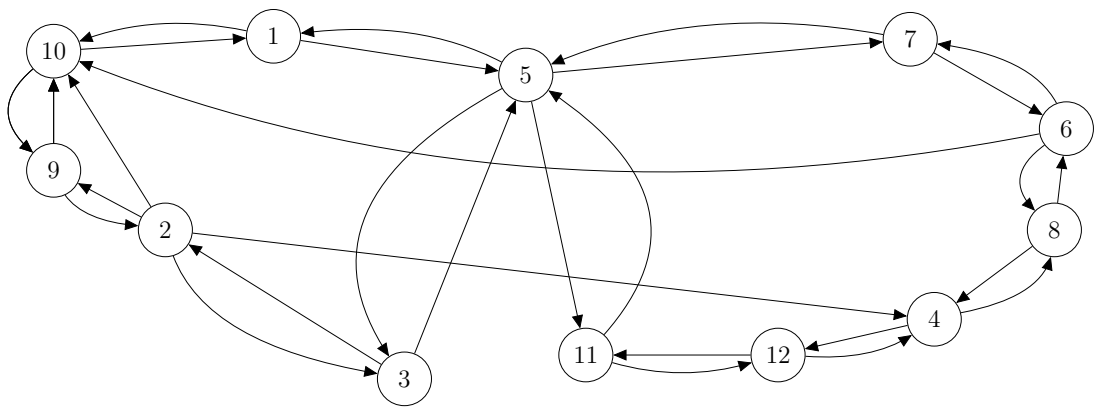
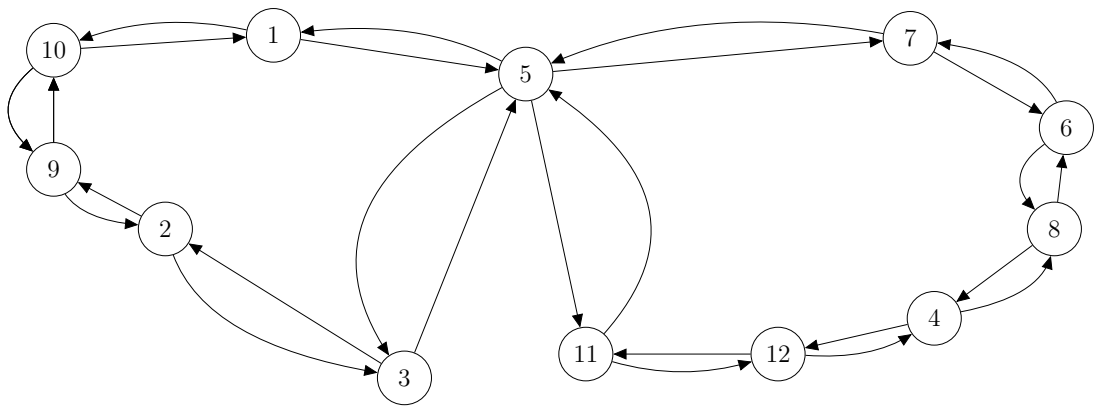
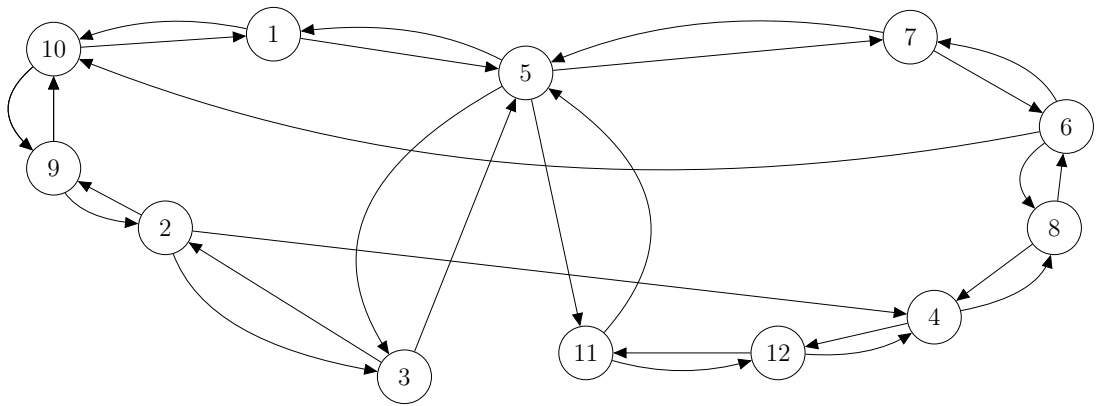

Let $G=(V,E)$ be a $k$-edge-connected directed graph. The problem of calculating a minimum size $k$-edge-connected spanning subgraph of $G$ is NP-hard for $k\geq 1$ \cite{G79,CT00}. Clearly, the MKESBSS problem is NP-hard for $k\geq 1$. Results of Edmonds \cite{Edmonds72} and Mader \cite{Mader85} imply that the number of edges in each minimal $k$-edge-connected directed graph is at most $2kn$ \cite{CT00}. Cheriyan and Thurimella \cite{CT00} gave a $(1+4/ \sqrt{k})$-approximation algorithm for the minimum $k$-edge-connected spanning subgraph problem. 
Moreover, Cheriyan and Thurimella \cite{CT00} provided a $(1+1/k)$-approximation algorithm for the minimum $k$-vertex-connected spanning subgraph problem. Georgiadis \cite{Georgiadis11} improved the running time of this algorithm for the minimum $2$-vertex-connected spanning subgraph (M2VCSS) problem and presented an efficient linear time approximation algorithm for the M2VCSS problem. Georgiadis et al. \cite{GIK20} presented efficient approximation algorithms based on the results of \cite{GT16,G10,ILS12,FILOS12} for the M2VCSS problem. Strongly connected components of a directed graph and blocks of an undirected graphs can be calculated in $O(n+m)$ time using Tarjan's algorithm \cite{TAARJAN72,Schmidt2013}.
Wu and Grumbach \cite{WG2010} introduced the concept of strongly biconnected directed graph and strongly biconnected components.  Italiano et al. \cite{ILS12,Italiano2010} provided linear time algorithms for computing all the strong articulation points and strong bridges of a directed graph. Georgiadis et al. \cite{GIK20} gave efficient approximation algorithms based on the results of \cite{GT16,G10,ILS12,FILOS12,FGILS2016} for the minimum $2$-vertex connected spanning subgraph problem. Jaberi \cite{Jaberi2019,Jaberi21,Jaberi2021} studied twinless articulation points and some related problems. Moreover, He studied b-bridges and some related problems \cite{Jaberi01897,Jaberi09793}. Georgiadis and Kosinas \cite{GeorgiadisandKosinas20} proved that twinless articulation points and twinless bridges can be found in $O(n+m)$ time. Jaberi studied minimum $2$-vertex strongly biconnected spanning directed subgraph problem in \cite{Jaberi47443} and
minimum $2$-vertex-twinless connected spanning subgraph problem in \cite{Jaberi03788}. 
 In this paper we study the minimum $2$-edge strongly biconnected spanning subgraph problem (denoted by M2ESBSS).

\section{Approximation algorithm for the Minimum $2$- edge strongly biconnected spanning directed subgraph problem} 

In this section we present an approximation algorithm  (Algorithm \ref{algo:approximationalgorithmfor2ebc}) for the M2ESBSS Problem.
\begin{lemma} \label{def:addingedgesfewersbcs}
Let $G=(V,E)$ be a strongly biconnected directed graph and let $U\subseteq E$ such that the directed subgraph $G_1=(V,U)$ is strongly connected but the underlying graph of $G_1$ is not biconnected. Let $(v,w)$ be an edge in $E\setminus U$  such that $v,w $ are not in the same strongly biconnected component of $G_1$. Then the directed subgraph $(V,U\cup \left\lbrace (v,w) \right\rbrace )$ has fewer  strongly biconnected components than $G_1$.
\end{lemma}
\begin{proof}
There is a simple path $p$ from $w$ to $v$ in $g_1$ since $G_1$ is strongly connected. Moreover, the edge $(v,w)$ and the path $p$ form a simple cycle in the subgraph $(V,U\cup \left\lbrace (u,w) \right\rbrace )$. Since the underlying graph of the subgraph $(V,U\cup \left\lbrace (u,w) \right\rbrace )$ contains a biconnected component that contains $v$ and $w$, the vertices $v$ and $w$ are in the same strongly biconnected component of the subgraph $(V,U\cup \left\lbrace (u,w) \right\rbrace )$. 
\end{proof}

\begin{figure}[htbp]
	\begin{myalgorithm}\label{algo:approximationalgorithmfor2ebc}\rm\quad\\[-5ex]
		\begin{tabbing}
			\quad\quad\=\quad\=\quad\=\quad\=\quad\=\quad\=\quad\=\quad\=\quad\=\kill
			\textbf{Input:} A $2$-edge strongly biconnected directed graph $G=(V,E)$ \\
			
			\textbf{Output:} a $2$-edge strongly biconnected subgraph of $G$\\
			{\small 1}\> identify a subset $U\subseteq E$ such that $H=(V,U)$ is a minimal\\
			{\small 2}\> $2$-edge-connected subgraph of $G$.\\
			{\small 3}\> \textbf{if} the subgraph $H=(V,U)$ is $2$-edge strongly biconnected \textbf{then} \\
			{\small 4}\>\>output  $H=(V,U)$\\
			{\small 5}\> \textbf{else}\\
			{\small 6}\>\>$E_{2e} \leftarrow U$\\
			{\small 7}\>\>  $G_{2e}\leftarrow (V,E_{2e})$ \\
			{\small 8}\>\> \textbf{while} the underlying graph of $G_{2e}$ is not biconnected \textbf{do}\\
            {\small 8}\>\>\>  compute the  strongly biconnected components of $G_{2e} $\\
		    {\small 11}\>\>\> find an edge $(v,w) \in E\setminus E_{2e}$ such that $v,w $ are not in \\
		    {\small 12}\>\>\>\>the same strongly biconnected components of $G_{2e} $.\\
		    {\small 13}\>\>\>\> add $(v,w)$ to $E_{2e}$. \\						
			{\small 14}\>\>  compute the b-bridges of $G_{2e}= (V,E_{2e})$.\\
			{\small 15}\>\>  \textbf{for} each b-bridge $t\in V$ \textbf{do} \\
			{\small 16}\>\>\>  \textbf{while} the underlying graph of $G_{2e}\setminus\left\lbrace t \right\rbrace $ is not   biconnected \textbf{do}\\ 
		{\small 17}\>\>\>\>  compute the  strongly biconnected components of $G_{2s}\setminus\left\lbrace t \right\rbrace $\\
		{\small 18}\>\>\>\> find an edge $(u,w) \in E\setminus E_{2e}$ such that $u,w $ are not in \\
		{\small 19}\>\>\>\>the same strongly biconnected components of $G_{2e}\setminus\left\lbrace t \right\rbrace $.\\
		{\small 20}\>\>\>\> add $(u,w)$ to $E_{2e}$. \\
		{\small 21}\>\>\>output $G_{2e}$
	
		\end{tabbing}
	\end{myalgorithm}
\end{figure}

\begin{lemma} 
The output of Algorithm \ref{algo:approximationalgorithmfor2ebc} is $2$-edge strongly biconnected.
\end{lemma}	
\begin{proof}
It follows from Lemma \ref{def:addingedgesfewersbcs}.
\end{proof}

The following lemma shows that any $2$-edge-strongly biconnected directed subgraph of a $2$-edge-strongly biconnected directed graph has at least $2n$ edges.
\begin{lemma} \label{def:optsolution2esb}
	Let $G=(V,E)$ be a $2$-edge-strongly biconnected directed graph. Let $U\subseteq E$ be an optimal solution for the M2ESBSS problem. Then the subgraph $(V,U)$ has at least $2n$ edges. 
	
\end{lemma}
\begin{proof}
the subgraph $(V,U)$ has no strong bridges since it has no b-bridges. Therefore, the subgraph $(V,U)$ is $2$-edge connected.
\end{proof}

Let $i$ be the number of b-bridges in  $H$. The following lemma shows that Algorithm \ref{algo:approximationalgorithmfor2ebc} has an approximation factor of $((5+i)/2)$.
\begin{theorem}
	Let $i$ be the number of b-bridges in $H$. Then, $|E_{2e}|\leq i(n-1)+5n$.
\end{theorem}
\begin{proof}
	Results of Edmonds \cite{Edmonds72} and Mader \cite{Mader85} imply that the number of edges in $H$ is at most $4n$ \cite{CT00}. Moreover, by Lemma \ref{def:optsolution2esb}, any $2$-edge-strongly biconnected directed subgraph of a $2$-edge-strongly biconnected directed graph has at least $2n$ edges.
 Algorithm \ref{algo:approximationalgorithmfor2ebc} adds at most $n-1$ edge to $E_{2e}$ in while loop of lines $8$--$13$ since the number of strong biconnected components of any directed graph is at most $n$.	
For every b-bridge in line $16$, Algorithm \ref{algo:approximationalgorithmfor2ebc} adds at most $n-1$ edge to $E_{2e}$ in while loop. Consequently, $|E_{2e}|\leq i(n-1)+5n$.
\end{proof}
\begin{Theorem}
	Algorithm \ref{algo:approximationalgorithmfor2ebc} runs in $O(n^{2}m)$ time.
\end{Theorem}
\begin{proof}
	A minimal $2$-edge-connected subgraph can be found in time $O(n^2)$ \cite{CT00}.
	The strongly biconnected components of a directed graph can be found in linear time \cite{WG2010}. 
	Moreover, by Lemma \ref{def:addingedgesfewersbcs}, lines $16$--$20$ take $O(nm)$ time.
\end{proof}

\section{Open Problems}
 Results of Mader \cite{Mader71,Mader72} imply that the number of edges in each minimal $k$-vertex-connected undirected graph is at most $kn$ \cite{CT00}. Results of Edmonds \cite{Edmonds72} and Mader \cite{Mader85} imply that the number of edges in each minimal $k$-vertex-connected directed graph is at most $2kn$ \cite{CT00}. Moreover, Results of Edmonds \cite{Edmonds72} and Mader \cite{Mader85} imply that the number of edges in each minimal $k$-edge-connected directed graph is at most $2kn$ \cite{CT00}.  Jaberi proved that each minimal $2$-vertex-strongly biconnected directed graph has at most $7n$ edges. The proof in is based on results of Mader \cite{Mader71,Mader72,Mader85} and Edmonds \cite{Edmonds72}.
 
 We leave as open problem whether the number of edges in each minimal $2$-edge strongly biconnected directed graph is at most $7n$ edges.

 Cheriyan and Thurimella \cite{CT00} gave a $(1+\lfloor 4/\sqrt{k} \rfloor\rfloor)$-approximation algorithm for the minimum $k$-edge-connected spanning subgraph problem for directed. For $k=2$, the second step of this algorithm can be modified in order to obtain a feasible solution for the M2ESBSS problem. An open problem is whether this algorithm has a good approximation factor for the M2ESBSS problem.

\end{document}